\documentclass[submission,copyright,creativecommons]{eptcs}
\providecommand{\event}{Non-Classical Logics 2024 (NCL'24)}

\usepackage{proof}
\usepackage{amssymb}
\usepackage{amsmath,amsthm}
\usepackage{color}

\newtheorem{theorem}{Theorem} 

\newtheorem{lemma}{Lemma}

\newtheorem{definition}{Definition}
\newtheorem{remark}[definition]{Remark}

\newtheorem{fact}[definition]{Fact}

\newcommand{\To}{{\to}}

\usepackage{iftex}

\ifpdf
  \usepackage{underscore}         
  \usepackage[T1]{fontenc}        
\else
  \usepackage{breakurl}           
\fi

\title{The Power of Generalized Clemens Semantics\thanks{The research of Hitoshi Omori was partially supported by a Sofja Kovalevskaja Award of the Alexander von Humboldt-Foundation, funded by the German Ministry for Education and Research. The research of Jonas R. B. Arenhart was also partially supported by the Alexander von Humboldt-Foundation through the Experienced Researcher Fellowship program, funded by CAPES-Humboldt.}}
\author{Hitoshi Omori
\institute{Graduate School of Information Sciences\\
Tohoku University\\
Sendai, Japan}
\email{hitoshiomori@gmail.com}
\and
Jonas R. B. Arenhart
\institute{Department of Philosophy\\
Federal University of Santa Catarina\\ 
Florian\'opolis, Brazil}
\email{jonas.becker2@gmail.com}
}
\def\titlerunning{The power of Generalized Clemens Semantics}
\def\authorrunning{H. Omori \& J.R.B.Arenhart}

\begin{document}
\maketitle

\begin{abstract}
In this paper, we elaborate on the ordered-pair semantics originally presented by Matthew Clemens for {\bf LP} (Priest's \textit{Logic of Paradox}). For this purpose, we build on a generalization of Clemens semantics to the case of $n$-tuple semantics, for every $n$. More concretely, i) we deal with the case of a language with quantifiers, and ii) we consider philosophical implications of the semantics. The latter includes, first, a reading of the semantics in epistemic terms, involving multiple agents. Furthermore, we discuss the proper understanding of many-valued logics, namely {\bf LP}  and {\bf K3} (Kleene strong $3$-valued logic), from the perspective of classical logic, along the lines suggested by Susan Haack. We will also discuss some applications of the semantics to issues related to informative contradictions, i.e. contradictions involving quantification over different respects a vague predicate may have, as advanced by Paul \'Egr\'e, and also to the mixed consequence relations, promoted by Pablo Cobreros, Paul \'Egr\'e, David Ripley and Robert van Rooij.  
\end{abstract}

\section{Introduction}

It goes without saying that the addition of truth values to the traditional pair consisting of `truth' and `falsity' brings several interesting technical consequences to the center of the stage. Understanding what such additional truth values mean, and how they affect the resulting logical system, however, constitutes a deep philosophical challenge. The intelligibility of such systems and, consequently, of their applications, both philosophical and in general, hangs on the prior understanding of such notions. 

The idea that there are difficulties related to the appropriate understanding of logical concepts is not new, although it has not always received the appropriate attention in the context of philosophical use of many-valued systems. In the literature about the subject, the topic has been forcefully discussed by Susan Haack in \cite[chap.11]{haack1978}. Haack advanced one specific proposal to achieve a clear picture of such systems. The first point of the proposal consists in preserving two-valuedness. As she puts it: 
\begin{quote}
    I think it is clear that a many-valued logic needn't require the admission of one or more extra truth-values over and above `true' and `false', and indeed, that it needn't even require the rejection of bivalence. \cite[p.213]{haack1978} 
\end{quote} 
The second step in the proposal advanced by Haack consists in offering an explanation of how, in a many-valued scenario, one will be able to retain two-valuedness (and, sometimes, even bivalence) and actually dispense with additional \textit{sui generis} truth values. The plan is actually quite simple: whatever seems to be prima facie an additional truth value should actually be explained away; it should be read in terms of the classical truth values and some additional epistemic or semantic ingredient that accounts for the distinct option in such settings. Discussing three-valued logics, for instance, Haack suggests that in some cases the third truth value may be understood as involving \textit{epistemic and/or semantic} restrictions to classical truth values, such as `true \textit{and known to be true by an agent} $O$', or `true \textit{and analytic}'; in other cases, its meaning needs not \textit{go beyond} the two truth values, such as when one is attributing to a proposition the value `neither true nor false', which is not an extra truth value. By building our understanding of many-valued systems solely on the already available two truth values, we use these previously intelligible notions of truth and falsity to endow such many-valued systems with the much needed clarification in terms of notions understood beforehand. 

In this paper, we adopt that general Haackian strategy as a means to increase intelligibility of a family of many-valued logics.\footnote{For further discussion of the Haackian strategy in a different context, see \cite{ISMVL22,OmoriArenhartTheoria}.} Our aim is to provide a clear meaning for some such logics by endowing our target systems with a common underlying ordered-pair semantics and generalizations of it. We shall discuss how close to retaining bivalence and two-valuedness such a strategy is. However, our overall claim is that such a semantics does contribute to advance the Haackian desideratum of granting understanding for many-valued logics. In particular, we discuss not only cases of Tarskian consequence relations, but also the notoriously obscure cases of mixed consequence relations. 

The rest of the paper is structured as follows. After briefly recalling the original semantics proposed by Clemens and a generalization within the propositional language in \S\ref{sec:Clemens}, we generalize the semantics for the language with quantifiers in \S\ref{sec:more-clemens}. We will also discuss the theme from Haack for the semantics proposed by Clemens. These will be followed by \S\ref{sec:reflections} in which we discuss three applications of the semantics. These include the issues related to the mixed consequence relations. Finally, the paper will be concluded with some brief remarks in \S\ref{sec:conclusion}.

\section{Revisiting ordered-pair semantics}\label{sec:Clemens}
In this section, we first briefly recall the original Clemens' semantics, as presented in \cite{Clemens}. We then go on to present a generalization as advanced in \cite{Lori2021}.

The language $\mathcal{L}_0$ consists of a set $\{ \neg , \land , \lor \}$ of propositional connectives and a countable set $\mathsf{Prop}$ of propositional variables which we denote by $p$, $q$, etc. Furthermore, we denote by $\mathsf{Form}$ the set of formulas defined as usual in $\mathcal{L}_0$. We denote a formula of $\mathcal{L}_0$ by $A$, $B$, $C$, etc. and a set of formulas of $\mathcal{L}_0$ by $\Gamma$, $\Delta$, $\Sigma$, etc.

We first revisit the ordered pair semantics as it was set out by Clemens.

\begin{definition} {\sloppy
A four-valued interpretation of $\mathcal{L}_0$ is a function $v$ from $\mathsf{Prop}$ to $\{ \langle 1, 1 \rangle,$ $\langle 1, 0 \rangle, \langle 0, 1 \rangle,  \langle 0, 0 \rangle \}$. Given a four-valued interpretation $v$, this is extended to a function $I: \mathsf{Form}\To\{ \langle 1, 1 \rangle,$ $\langle 1, 0 \rangle, \langle 0, 1 \rangle,  \langle 0, 0 \rangle \}$ as follows}:

\begin{center}
{\small
\begin{tabular}{cccc}
\begin{tabular}{c|c}
       & $\neg$    \\ \hline
$\langle 1, 1 \rangle$ & $\langle 0, 0 \rangle$  \\
$\langle 1, 0 \rangle$ & $\langle 0, 1 \rangle$\\
$\langle 0, 1 \rangle$ & $\langle 1, 0 \rangle$\\
$\langle 0, 0 \rangle$ & $\langle 1, 1 \rangle$  \\  
\end{tabular}
\quad &
\begin{tabular}{c|ccccccc}
$\land$ & $\langle 1, 1 \rangle$ & $\langle 1, 0 \rangle$ & $\langle 0, 1 \rangle$ & $\langle 0, 0 \rangle$ \\   \hline
$\langle 1, 1 \rangle$ & $\langle 1, 1 \rangle$ & $\langle 1, 0 \rangle$ & $\langle 0, 1 \rangle$ & $\langle 0, 0 \rangle$ \\ 
$\langle 1, 0 \rangle$ & $\langle 1, 0 \rangle$ & $\langle 1, 0 \rangle$ & $\langle 0, 1 \rangle$ & $\langle 0, 0 \rangle$ \\  
$\langle 0, 1 \rangle$ & $\langle 0, 1 \rangle$ & $\langle 0, 1 \rangle$ & $\langle 0, 1 \rangle$ & $\langle 0, 0 \rangle$ \\  
$\langle 0, 0 \rangle$ & $\langle 0, 0 \rangle$ & $\langle 0, 0 \rangle$ & $\langle 0, 0 \rangle$ & $\langle 0, 0 \rangle$ \\  
\end{tabular}
\quad &
\begin{tabular}{c|ccccccc}
$\lor$ & $\langle 1, 1 \rangle$ & $\langle 1, 0 \rangle$ & $\langle 0, 1 \rangle$ & $\langle 0, 0 \rangle$ \\   \hline
$\langle 1, 1 \rangle$ & $\langle 1, 1 \rangle$ & $\langle 1, 1 \rangle$ & $\langle 1, 1 \rangle$ & $\langle 1, 1 \rangle$ \\ 
$\langle 1, 0 \rangle$ & $\langle 1, 1 \rangle$ & $\langle 1, 0 \rangle$ & $\langle 1, 0 \rangle$ & $\langle 1, 0 \rangle$ \\  
$\langle 0, 1 \rangle$ & $\langle 1, 1 \rangle$ & $\langle 1, 0 \rangle$ & $\langle 0, 1 \rangle$ & $\langle 0, 1 \rangle$ \\  
$\langle 0, 0 \rangle$ & $\langle 1, 1 \rangle$ & $\langle 1, 0 \rangle$ & $\langle 0, 1 \rangle$ & $\langle 0, 0 \rangle$ \\  
\end{tabular}
\end{tabular}
}
\end{center}
\end{definition}

\begin{remark}
Note that truth tables for conjunction and disjunction result from adapting $\min$ and $\max$ definitions, respectively, with the following order on the values: $\langle 0, 0 \rangle {<}\langle 0, 1 \rangle {<} \langle 1, 0 \rangle {<} \langle 1, 1 \rangle$. 
\end{remark}

\begin{definition}
For all $\Gamma\cup \{ A \}\subseteq \mathsf{Form}$, $\Gamma\models_3 A$ iff for all four-valued interpretations $v$, $I(A)\in \mathcal{D}$ if $I(B)\in \mathcal{D}$ for all $B\in \Gamma$, where $\mathcal{D}=\{ \langle 1, 1 \rangle, \langle 1, 0 \rangle, \langle 0, 1 \rangle \}$.\footnote{We are using the subscript 3 just to indicate that there are three designated values. See \S\ref{sub3.2} for details concerning the intended reading of such truth values.}
\end{definition}

We now recall the standard three-valued semantics for {\bf LP} (see \cite[\S7.4]{Priest08}).

\begin{definition}
A three-valued interpretation of $\mathcal{L}_0$ is a function $v_3: \mathsf{Prop}\To \{ \mathbf{t}, \mathbf{i}, \mathbf{f} \}$. Given a three-valued interpretation $v_3$, this is extended to a function $I_3$ from $\mathsf{Form}$ to $\{ \mathbf{t}, \mathbf{i}, \mathbf{f} \}$ by truth functions depicted in the form of truth tables as follows:

\begin{center}
{\small
\begin{tabular}{ccc}
\begin{tabular}{c|c}
       & $\neg$    \\ \hline
$\mathbf{t}$ & $\mathbf{f}$  \\
$\mathbf{i}$ & $\mathbf{i}$\\
$\mathbf{f}$ & $\mathbf{t}$  \\  
\end{tabular}
\quad &
\begin{tabular}{c|cccccc}
$\land$ & $\mathbf{t}$ & $\mathbf{i}$ & $\mathbf{f}$ \\   \hline
$\mathbf{t}$ & $\mathbf{t}$ & $\mathbf{i}$ & $\mathbf{f}$ \\ 
$\mathbf{i}$ & $\mathbf{i}$ & $\mathbf{i}$ & $\mathbf{f}$ \\  
$\mathbf{f}$ & $\mathbf{f}$ & $\mathbf{f}$ & $\mathbf{f}$ \\  
\end{tabular}
\quad &
\begin{tabular}{c|cccccc}
$\lor$ & $\mathbf{t}$ & $\mathbf{i}$ & $\mathbf{f}$ \\   \hline
$\mathbf{t}$ & $\mathbf{t}$ & $\mathbf{t}$ & $\mathbf{t}$ \\ 
$\mathbf{i}$ & $\mathbf{t}$ & $\mathbf{i}$ & $\mathbf{i}$ \\  
$\mathbf{f}$ & $\mathbf{t}$ & $\mathbf{i}$ & $\mathbf{f}$ \\  
\end{tabular}
\end{tabular}
}
\end{center}
\end{definition}

\begin{definition}\label{def:LP-cons}
For all $\Gamma\cup \{ A \}\subseteq \mathsf{Form}$, $\Gamma\models_{\bf LP} A$ iff for all three-valued interpretations $v_3$, $I_3(A)\in \mathcal{D}$ if $I_3(B)\in \mathcal{D}$ for all $B\in \Gamma$, where $\mathcal{D}=\{ \mathbf{t}, \mathbf{i} \}$.
\end{definition}

Based on these, Clemens established the following result in \cite{Clemens}.

\begin{fact}[Clemens]
For all $\Gamma\cup \{ A \}\subseteq \mathsf{Form}$, $\Gamma\models_3 A$ iff $\Gamma\models_{\bf LP} A$.
\end{fact}

Let us now present a generalization of the above result, presented in \cite{Lori2021}. For this purpose, we refer to the two-element Boolean algebra as {\bf 2}. 

\begin{definition}\label{def:nCp2}
For $n\geq 2$, we define ${\bf 2}^n$ as the $n$-ary Cartesian product of {\bf 2} with the lexicographical order. Given $\langle x_1, \dots , x_n \rangle\in {\bf 2}^n$, we define a unary operation $-: {\bf 2}^n \to {\bf 2}^n$ as follows: $-\langle x_1, \dots , x_n \rangle:= \langle 1-x_1, \dots , 1-x_n \rangle$.
\end{definition}

\begin{definition}
An $n$-interpretation of $\mathcal{L}_0$ is a function $v: \mathsf{Prop}\To {\bf 2}^n$. Given an $n$-interpretation $v$, this is extended to a function $I: \mathsf{Form}\To {\bf 2}^n$ as follows: $I(p)=v(p)$; $I(\neg A){=}{-}I(A)$; $I(A{\land} B){=}\min (I(A), I(B))$; $I(A{\lor} B){=}\max (I(A), I(B))$.
\end{definition}

\begin{definition}
For all $\Gamma\cup \{ A \}\subseteq \mathsf{Form}$, $\Gamma\models_{n,t} A$ (tolerant consequence based on $n$-interpretations) iff for all $n$-interpretations $v$, $I(A)\in \mathcal{D}$ if $I(B)\in \mathcal{D}$ for all $B\in \Gamma$, where $\mathcal{D}={\bf 2}^n\setminus \{ \langle 0, 0, \dots, 0 \rangle \}$.
\end{definition}

We are now ready to recall a generalization of Clemens' observation. 
\begin{theorem}
For all $\Gamma\cup \{ A \}\subseteq \mathsf{Form}$, $\Gamma\models_{n,t} A$ iff $\Gamma\models_{\bf LP} A$.
\end{theorem}

Although the motivation of Clemens in \cite{Clemens} was focused exclusively on {\bf LP}, with a special emphasis on negation, we obtain an interesting insight into the relation between {\bf LP} and other related systems, namely {\bf CL} (classical logic) and {\bf K3}. For the purpose of clarifying our point, however, we shall state the results in the language with quantifiers, and this is the goal of the next section.\footnote{For a discussion of the propositional cases, see \cite{Lori2021}.}

\section{More on ordered-pair semantics}\label{sec:more-clemens}

\subsection{Basic observations}
The language $\mathcal{L}_1$ consists of the following vocabulary: a set $\{ \neg , \land , \lor \}$ of propositional connectives, the universal and particular quantifiers $\forall$ and $\exists$, a countable set $\{ x_0, x_1,\dots \}$ of variables, a countable set $\{c_0, c_1, \dots \}$ of constant symbols, and a countable set $\{ P_0, P_1, \dots \}$ of predicate symbols, where we associate each predicate $P_{k}$ with a fixed finite arity. We regard $0$-ary predicate symbols as propositional letters. We define the set of formulas in $\mathcal{L}_1$ as follows: 
\[ A ::= P(t_{1},\dots,t_{n}) \,|\, \neg A \,|\, A \land B \,|\,  A \lor B \,|\,  \forall x A \,|\,  \exists x A, \]
where $t_{i}$ is a {\em{term}}, namely a variable or a constant symbol. 
We say that a formula is {\em propositional} if it is constructed from propositional letters (i.e., $0$-ary predicate symbols) by using the propositional connectives. We define the notions of {\em free} and {\em bound} variable, and {\em sentence} as usual.  
We write $A_x(t)$ to mean the result of substituting all the occurrences of free variable $x$ in $A$ by the term $t$, renaming the bound variables, if necessary, to avoid variable-clashes. We denote sets of formulas by $\Gamma$, $\Sigma$, etc.

Let us first recall the three-valued semantics for {\bf K3} and {\bf LP}.

\begin{definition}
A {\em three-valued interpretation} $\mathcal{I}$ for $\mathcal{L}_1$ is a pair $\langle D, v \rangle$ where $D$ is a non-empty set and we assign $v(c) \in D$ to each constant $c$, assign an $i$-place function $v(P): D^{i} \longrightarrow \{ 0, 1/2, 1 \}$ to each $i$-ary predicate symbol $P$. Given any interpretation $\langle D, v \rangle$, we can define Clemens-valuation $\overline{v}$ for all the sentences of $\mathcal{L}_1$ expanded by $\{ k_d : d \in D \}$ inductively as follows: as for the atomic {\em sentences}, 

\begin{itemize}
\item $\overline{v}(P(t_{1},...,t_{n}))=v(P)(v(t_{1}),\dots,v(t_{n})).$
\end{itemize}

\noindent 
The rest of the clauses are as follows:

\begin{itemize}
\item $\overline{v}(\neg A)=1-\overline{v}(A)$
\item $\overline{v}(A \land B)=\min (\overline{v}(A), \overline{v}(B))$
\item $\overline{v}(A \lor B)=\max (\overline{v}(A), \overline{v}(B))$
\item $\overline{v}(\forall x A)= \min (\{ \overline{v}(A_x(k_d)): d \in D \})$
\item $\overline{v}(\exists x A)= \max (\{ \overline{v}(A_x(k_d)): d \in D \})$
\end{itemize}
\end{definition}

\begin{definition}
For all sets of sentences $\Gamma\cup \{ A \}$, $\Gamma\models_i A$ iff for all three-valued interpretations $\mathcal{I}$, $\overline{v}(A)\in \mathcal{D}_i$ if $\overline{v}(B)\in \mathcal{D}_i$ for all $B\in \Gamma$, where $i\in \{ k, l \}$ and 
\begin{itemize}
\item $\mathcal{D}_k=\{ 1  \}$ ($k$ for {\bf K3}), 
\item $\mathcal{D}_l=\{ 1, 1/2 \}$ ($l$ for {\bf LP}).
\end{itemize}
\end{definition}

Moreover, building on the notation above, we introduce an instance of the $p$-consequence relation (cf. \cite{Frankowski2004formalization}), as follows.

\begin{definition}
For all sets of sentences $\Gamma\cup \{ A \}$, $\Gamma\models_{st} A$ iff for all three-valued interpretations $\mathcal{I}$, $\overline{v}(A)\in \mathcal{D}_l$ if $\overline{v}(B)\in \mathcal{D}_k$ for all $B\in \Gamma$.
\end{definition}

Finally, we refer to the semantic consequence relation based on the standard two-valued interpretations for classical logic as $\models_2$.

We now turn to introduce the semantics inspired by Clemens.

\begin{definition}
A {\em Clemens interpretation} $\mathcal{I}$ for $\mathcal{L}_1$ is a pair $\langle D, v \rangle$ where $D$ is a non-empty set and we assign $v(c) \in D$ to each constant $c$, assign an $i$-place function $v(P): D^{i} \longrightarrow {\bf 2}^n$ to each $i$-ary predicate symbol $P$. Given any interpretation $\langle D, v \rangle$, we can define Clemens-valuation $\overline{v}$ for all the sentences of $\mathcal{L}_1$ expanded by $\{ k_d : d \in D \}$ inductively as follows: as for the atomic {\em sentences}, 
\[
\overline{v}(P(t_{1},...,t_{n}))=v(P)(v(t_{1}),\dots,v(t_{n})).
\]
The rest of the clauses are as follows (recall Definition~\ref{def:nCp2}):
\begin{itemize}
\item $\overline{v}(\neg A)={-}\overline{v}(A)$
\item $\overline{v}(A \land B)=\min (\overline{v}(A), \overline{v}(B))$
\item $\overline{v}(A \lor B)=\max (\overline{v}(A), \overline{v}(B))$
\item $\overline{v}(\forall x A)= \min (\{ \overline{v}(A_x(k_d)): d \in D \})$
\item $\overline{v}(\exists x A)= \max (\{ \overline{v}(A_x(k_d)): d \in D \})$
\end{itemize}
\end{definition}

\begin{definition}
For all set of sentences $\Gamma\cup \{ A \}$, $\Gamma\models_{n,i} A$ iff for all Clemens interpretations $\mathcal{I}$, $\overline{v}(A)\in \mathcal{D}_i$ if $\overline{v}(B)\in \mathcal{D}_i$ for all $B\in \Gamma$, where $i\in \{ s, b, t \}$ and \begin{itemize}
\item $\mathcal{D}_s=\{ \langle 1, 1, \dots, 1 \rangle \}$ ($s$ for strict), \item $\mathcal{D}_b=\{ \langle 1, x_2, \dots, x_n \rangle :  x_2, \dots, x_n \in {\bf 2} \}$ ($b$ for bossy), and 
\item $\mathcal{D}_t={\bf 2}^n\setminus \{ \langle 0, 0, \dots, 0 \rangle \}$ ($t$ for tolerant).
\end{itemize}
\end{definition}

Moreover, building on the notation above, we introduce another instance of the $p$-consequence relation as follows.

\begin{definition}
For all set of sentences $\Gamma\cup \{ A \}$, $\Gamma\models_{n,s,t} A$ iff for all Clemens interpretations $\mathcal{I}$, $\overline{v}(A)\in \mathcal{D}_t$ if $\overline{v}(B)\in \mathcal{D}_s$ for all $B\in \Gamma$.
\end{definition}

In what follows, we will establish the equivalence of consequence relations based on the two semantics.

\begin{lemma}\label{lem:cto3}
Given a Clemens-interpretation $\langle D, v \rangle$, define the three-valued interpretation $\langle D', v' \rangle$ as follows:
\begin{itemize}
\item $D':= D$
\item For each constant $c$, $v'(c){:=} v(c)$ and for each $i$-ary predicate symbol $P$, 

\begin{itemize}
\item $v'(P)(d_1, \dots d_i){=}1$ if  $v(P)(d_1, \dots d_i)=\langle 1, 1, \dots, 1 \rangle$
\item $v'(P)(d_1, \dots d_i){=}1/2$ if  $v(P)(d_1, \dots d_i)\in {\bf 2}^n\setminus \{ \langle 1, 1, \dots, 1 \rangle, \langle 0, 0, \dots, 0 \rangle \}$
\item $v'(P)(d_1, \dots d_i){=}0$ if  $v(P)(d_1, \dots d_i)=\langle 0, 0, \dots, 0 \rangle$
\end{itemize}
\end{itemize}
\noindent Then, for all sentences $A$, \textnormal{(a)} $\overline{v'}(A){=}1$ iff $\overline{v}(A){=}\langle 1, 1, \dots, 1 \rangle$; \textnormal{(b)} $\overline{v'}(A){=}0$ iff $\overline{v}(A){=}\langle 0, 0, \dots, 0 \rangle$.
\end{lemma}

\begin{proof}
By induction on the complexity of $A$. 
\end{proof}

\begin{lemma}\label{lem:3toc}
Given a three-valued interpretation $\langle D, v \rangle$, define the Clemens interpretation $\langle D', v' \rangle$ as follows:
\begin{itemize}
\item $D':= D$
\item For each constant $c$, $v'(c){:=} v(c)$ and for each $i$-ary predicate symbol $P$, 

\begin{itemize}
\item $v'(P)(d_1, \dots d_i){=}\langle 1, 1, \dots, 1, 1 \rangle$ if  $v(P)(d_1, \dots d_i){=}1$
\item $v'(P)(d_1, \dots d_i){=}\langle 1, 1, \dots, 1, 0 \rangle$ if  $v(P)(d_1, \dots d_i){=}1/2$
\item $v'(P)(d_1, \dots d_i){=}\langle 0, 0, \dots, 0, 0 \rangle$ if  $v(P)(d_1, \dots d_i){=}0$
\end{itemize}
\end{itemize}
\noindent Then, for all sentences $A$, \textnormal{(a)} $\overline{v'}(A){=}\langle 1, 1, \dots, 1 \rangle$ iff $\overline{v}(A){=}1$; \textnormal{(b)} $\overline{v'}(A){=}\langle 0, 0, \dots, 0 \rangle$ iff $\overline{v}(A){=}0$.
\end{lemma}

\begin{proof}
By induction on the complexity of $A$. 
\end{proof}

\begin{theorem}
For all set of sentences $\Gamma\cup \{ A \}$, 
\textnormal{(i)} $\Gamma\models_{n,s} A$ iff $\Gamma\models_k A$,
\textnormal{(ii)} $\Gamma\models_{n,b} A$ iff $\Gamma\models_2 A$, and
\textnormal{(iii)} $\Gamma\models_{n,t} A$ iff $\Gamma\models_l A$.\footnote{Note that one can also obtain similar results by building on the framework due to Hans Herzberger, presented in \cite{herzberger1973dimensions}. For some discussions related to the results, cf. \cite{OmoriArenhartTheoria}.    }
\end{theorem}
\begin{proof}
Ad. (i): For the right-to-left direction, suppose $\Gamma\not\models_{n,s} A$. Then, there is a Clemens-interpretation $\mathcal{I}$ such that $\overline{v}(B){=}\langle 1, 1, \dots, 1 \rangle$ for all $B{\in} \Gamma$ and $\overline{v}(A){\neq} \langle 1, 1, \dots, 1 \rangle$. By making use of (a) of Lemma~\ref{lem:cto3}, there is a three-valued interpretation $\mathcal{I}'=\langle D', v'\rangle$ such that we obtain that $\overline{v'}(B){=} 1$ for all $B{\in}\Gamma$ and $\overline{v'}(A){\neq}1$, that is $\Gamma\not\models_k A$. For the other way around, suppose $\Gamma\not\models_k A$. Then, there is a three-valued interpretation $\mathcal{I}$ such that $\overline{v}(B){=}1$ for all $B{\in} \Gamma$ and $\overline{v}(A){\neq} 1$. By making use of (a) of Lemma~\ref{lem:3toc}, there is a Clemens-valued interpretation $\mathcal{I}'=\langle D', v'\rangle$ such that we obtain that $\overline{v'}(B){=} \langle 1, 1, \dots, 1, 1 \rangle$ for all $B{\in}\Gamma$ and $\overline{v'}(A){\neq} \langle 1, 1, \dots, 1, 1 \rangle$, that is $\Gamma\not\models_{n,s} A$. 

Ad (ii): The proof runs in the above manner, but we make use of lemmas that are obtained by making some obvious modifications to Lemmas~\ref{lem:cto3} and \ref{lem:3toc}.

Ad (iii): The proof again runs in the above manner, but we make use of (b), instead of (a), of Lemmas~\ref{lem:cto3} and \ref{lem:3toc}. This completes the proof. 
\end{proof}

\begin{theorem}
For all set of sentences $\Gamma\cup \{ A \}$, $\Gamma\models_{n,s,t} A$ iff $\Gamma\models_{st} A$.
\end{theorem}

\begin{proof}
Suppose $\Gamma\not\models_{n,s,t} A$. Then, there is a Clemens-interpretation $\mathcal{I}$ such that $\overline{v}(B){=}\langle 1, 1, \dots, 1 \rangle$ for all $B{\in} \Gamma$ and $\overline{v}(A){=} \langle 0, 0, \dots, 0 \rangle$. By making use of Lemma~\ref{lem:cto3}, there is a three-valued interpretation $\mathcal{I}'=\langle D', v'\rangle$ such that we obtain that $\overline{v'}(B){=} 1$ for all $B{\in}\Gamma$ and $\overline{v'}(A){=}0$, that is $\Gamma\not\models_{st} A$. For the other way around, suppose $\Gamma\not\models_{st} A$. Then, there is a three-valued interpretation $\mathcal{I}$ such that $\overline{v}(B){=}1$ for all $B{\in} \Gamma$ and $\overline{v}(A){=} 0$. By making use of Lemma~\ref{lem:3toc}, there is a Clemens-valued interpretation $\mathcal{I}'=\langle D', v'\rangle$ such that we obtain that $\overline{v'}(B){=} \langle 1, 1, \dots, 1, 1 \rangle$ for all $B{\in}\Gamma$ and $\overline{v'}(A){=} \langle 0, 0, \dots, 0, 0 \rangle$, that is $\Gamma\not\models_{n,s,t} A$. 
\end{proof}

\subsection{Clemens in view of Haack}\label{sub3.2}

Now that the basics of the generalized and first-order Clemens semantics is presented, we may return to the problem of providing for understanding of many-valued logics, as raised by Susan Haack, in view of the Clemens semantics. More explicitly, we need to address how the Clemens semantics contributes to fulfil the explicit demand for intelligibility advanced by Haack. As we have commented in the introduction, Haack's strategy for the understanding of some prima facie candidate for a \textit{sui generis} truth value is as follows: in order to endow a system of many-valued logic with intelligibility, we should attempt to `read' such truth values in terms of the already known and understood classical truth values, possibly with additional semantic or epistemic contours. If that can be done, the need for additional truth values is actually avoided, and we have explained them away, in a sense.

Given that demand, the next natural question is: can one such `classical reading' of the truth values be attributed to the generalization of the semantics advanced by Clemens? It is our contention now that this is perfectly possible, and more, that the framework presented is quite classical, in a sense. Let us focus on the simple ordered-pair semantics as originally presented by Clemens, where the set of truth values is $\{ \langle 1, 1 \rangle,$ $\langle 1, 0 \rangle, \langle 0, 1 \rangle,  \langle 0, 0 \rangle \}$ (it is a simple matter to extend the readings to more general cases). Clearly, given the order established for the truth values, and the division between designated and non-designated truth values in order to define the classical consequence relation, it is not difficult to see \emph{the first component} of the pairs as playing a more prominent role than all the others. That is, there is a natural reading of the truth values where the first component marks a division between truth and untruth, regardless of what the second component adds to it. 

One may consider this reading favoring the first component as being elaborated according to a kind of realist approach to truth and falsity that classical logic is said to promote anyway. There is a sense in which propositions in a classical setting are defined as to their truth or falsity independently of whether any one agent knows the relevant facts about such a distribution of truth values (a discussion is to be found in \cite{dummett1984}). In this kind of reading, the classical meaning of the connectives may be properly understood as available through the first component of the truth values, while the other components add an epistemic dimensions, that is, they add, and that is one possibility, appreciations of different agents that may disagree on the truth value of some proposition. A first shot on understanding what is going on, suggested by Clemens himself (\cite[p.202]{Clemens}), advanced the readings as follows:  

\smallskip

\begin{tabular}{ll}
$\langle 1, 1 \rangle$ = true, and true only;     & $\langle 0, 1 \rangle$ = false, but also true; \\
$\langle 1, 0 \rangle$ = true, but also false;     & $\langle 0, 0 \rangle$ = false, and false only. 
\end{tabular}

\smallskip

\noindent That makes for an interesting first attempt in the direction of a better understanding of the truth values involved in terms of the already available classical truth values, conferring also classical intelligibility to {\bf K3} and {\bf LP}. It is clear that reading the truth values like that requires that some sentences receive two of such classical truth values some times, but that is not a problem; as Haack comments on what concerns the case of truth-value gaps in \cite[p.213]{haack1978}:
\begin{quote}
Assignment of the third truth value to a wff [well-formed formula] indicates that it has \textit{no} truth value, not that it has a non-standard, third truth value.   
\end{quote}
So, in the case of gaps, intelligibility is preserved. By parity of reasoning, of course, attributing two classical truth values to a formula is also not the attribution of a non-standard truth value; it is merely attribution of two of the available truth values to a formula (for further discussions, see \cite{Lori2021}).

As a result, in the sense required by Haackian demands of intelligibility, the semantics presented by Clemens does seem to stay very close to a classical semantics, allowing for readings in terms of truth and falsity that stay very close to classical logic. The division between two groups of truth and false sentences then contribute to the idea that no additional \textit{sui generis} truth values has been added.

We have also seen that one may generalize the semantics from ordered pairs to $n$-tuples of classical truth values. In a sense, that causes no additional complication on what concerns the understanding of such truth values. The order attributed to the $n$-tuples does the work in getting the appropriate division between truths and falsehoods when it comes to obtain classical logic: truth is whatever has truth as its first component; false is whatever has falsity as its first component. 

In the cases of non-classical systems --- {\bf K3} and {\bf LP}--- the motivations for the choice of designated truth values may come from different fronts, regarding the role of the order of the truth values and additional suppositions that a more nuanced choice may be motivated. One may see {\bf K3} as involving choice of certified or strict truth as designated, while {\bf LP} may be seen as involving a tolerant approach where only certified falsity is excluded. 

The role of the order of truth values will play a prominent role in selecting each kind of systems available, and providing for interesting uses of such systems. In the next section, we shall follow Haack and suggest more epistemic-oriented readings of the truth values

\section{Reflections}\label{sec:reflections}

In this section, we discuss various topics concerning the proper understanding of the semantics proposed here, and how the Haackian demand for intelligibility may be achieved by using such a semantics. In particular, we focus on how the framework developed here sheds light on some not so clear issues concerning many-valued logics and their applications in connection with mixed consequence relations. 

\subsection{Topic I - Agent reading}
To begin with, besides the original Clemens reading, we introduce one additional possible reading for the truth values available in the generalization of Clemens semantics. Doing so will offer a more epistemic reading of the truth values, which justifies our claim that we are following the Haackian strategy of reading typical additional truth values in terms of the classical truth values with epistemic restrictions on them. Such reading also motivates a plainly classical understanding of the different consequence relations available, as well as motivates a discussion on the meaning of the connectives (again, the reader may also see the discussion in \cite{Lori2021,OmoriArenhartTheoria}).\footnote{Haack did not, in fact, extend her discussion of the understanding of the additional truth values to the consequence relations in scenarios involving such truth values.} We can think of roughly the following kind of intuitive reading that is seen as the result of epistemic qualification to the classical truth values. 

A specifically epistemic reading may be conferred to the order of the truth values in the $n$-tuples available for the Clemens semantics if we approach it in terms of $n$ distinct agents, each of whom is supposed to evaluate the classical truth value of any given atomic proposition. Each element of an $n$-tuple then corresponds to the evaluation of the $n$-th agent. In order to make sense of the order of truth values, we can rank agents confidence in their evaluation too, so that we may think of going from specialists on a topic --- the first entry from left to right --- to someone who is not actually specialist on the topic ---the first entry from right to left. Collectively, once the evaluations are performed for the atomic propositions, we may compute the truth values of complex propositions by evaluating the Boolean connectives. 

Besides including agents, one can also think of an epistemic reading that is less focused on human beings, and more focused on procedures, reading the positions on the truth values as different tests that may be applied to check the application of a given predicate, with tests varying on their rigour or confidence. So, with $n$ tests, an $n$-tuple would fill with $1$ or $0$ the $n$-th position depending on whether the $n$ test is positive or negative for the application of the predicate (this reading is a generalization and adaptation of a discussion by Newton da Costa in \cite[p.131]{dacostaensaio}).\footnote{The first edition of da Costa's book is from 1980.} The order of the truth values would rank the degree of confidence we accept for a test in granting that the predicate does apply. 

By focusing on the agent reading for the sake of simplicity, the distinct consequence relations should be read: 
\begin{enumerate} 
\item {\bf K3} is the logic resulting from preserving only what all the agents agree on being true. In this sense, this logic requires unanimity if a proposition is to follow from a unanimous set of premises; 
\item {\bf CL} requires that the first agent should be seen as having privileged epistemic abilities, so that validity is related to whatever that particular agent judges as true. One may also consider that the first component is a kind of God's eye point of view, never failing, and the other ones are fallible human beings; 
\item {\bf LP}  results when one is more tolerant towards all of the agents opinions; validity is prevented only in cases where there is consensus about falsity. 
\end{enumerate}

\subsection{Topic II: Respects}

Another interesting application of the generalized Clemens semantics may be found in relation to Paul \'Egr\'e's discussion on acceptable contradictions in \cite{egre2019} (we omit some of the more fine grained details of \'Egr\'e's exposition). The discussion is related to the dialetheist claim that some contradictions may be actually true (and also false). The major example of one such contradiction, of course, is derived from discussions on the Liar paradox (see \cite{Priest2006-PRIICA} for the \textit{locus classicus}). In one possible presentation, the Liar sentence may be presented by introducing a sentence $\lambda$ that says of itself that it is false: 
\begin{itemize}
    \item[$\lambda:$] The sentence $\lambda$ is false. 
\end{itemize}
With very simple logical derivations usually available, one may then derive that the Liar sentence is both true and false. 

According to \'Egr\'e, if one is going to accept some contradictions, as a dialetheist is motivated to, one should attempt to make clear sense of such contradictions. In particular, it may happen, as dialetheists argue, that some contradictions are actually informative, not empty of content. \'Egr\'e then goes on to define an \textit{acceptable contradiction} as an informative sentence of the form `$x$ is $P$ and $x$ is not $P$'. The contradictions are understood as involving a kind of vagueness, they hide some additional information regarding the assertion of a predicate and its denial; literally, one is asserting the predicate according to some regards or respects, while at the same time denying it according to other respects: 
\begin{quote}
the acceptability of contradictions involving adjectives in particular (including ``true'') might indeed be grounded in the availability of multiple respects of application, but provided those respects of comparison are closely related to each other in a way that is constitutive of the vagueness of the expression in question. \cite[p.41]{egre2019}
\end{quote}
This looks quite similar to the above criteria of application of a predicate; the same predicate had to have different criteria of application, which could result in different verdicts concerning the appropriateness of application of the predicate. Now, instead of criteria of application, what we have is different respects associated with the same predicate. Contradictory sentences involve quantification over respects available for the application of the terms that are involved in generating the contradiction. This may be the case for adjectives, like `good', `intelligent', `tall'. A contradiction like 
\begin{itemize}
\item John is rich and John is not rich
\end{itemize}
is then understood as involving quantification over respects, with the latter indicating that John may be rich in some respects, but not rich in other (different) respects. For example, it may be that, in regard of academic professors, John is actually rich, while, at the same time, according to the standard used to compute latest list of billionaires in the world, John is not even close to being rich. 

The idea that acceptable contradictions involve quantification over respects applies not only to adjectives, but also to nouns, for example: 
\begin{itemize}
  \item Mario is a man and is not a man,  
\end{itemize}
where the first occurrence of `man' designates `man with respect to gender', and the second one designates `man with respect to satisfaction of some stereotype of masculinity'. We may quantify over respects also in the case of verbs: 
\begin{itemize}
\item I like fish and I don't like fish    
\end{itemize}
which indicates that there are some respects according to which I like fish, let us say, as animals, while it also indicates that I don't like fish in every respect, let us say, as an option for a meal. The plan, remember, is that ``each time contradictions can be paraphrased by means of an explicit specification of distinct respects of application.'' (\cite[p.44]{egre2019})

As a template of the analysis of informative contradictions in terms of the proposed paraphrase using different respects, we have the following scheme: 
\begin{itemize}
\item $x$ is $P$ [in some respects], and $x$ is not $P$ [in some respects].
\end{itemize}
\'Egr\'e prefers the following way of putting it (this will be relevant for us soon):
\begin{itemize}
    \item $x$ is $P$ [in some respects], and $x$ is not [in all respects] $P$.
\end{itemize}
Given this account of informative contradictions, the informativeness is accounted for by the fact that ``the respects relevant to the second conjunct are distinct from the respects relevant to the first'' (\cite[p.46]{egre2019}). That means that it is different information that is being dealt with in the affirmation and in the negation. 

In summary again, the plan is the following:
\begin{quote}
    The basic idea is that relevant respects determine different extents to which a property can be satisfied, and those extents can be quantified over. \cite[p.50]{egre2019}
\end{quote}
This availability of different respects for application of a predicate opens the door for application of Clemens semantics. Using the generalized Clemens semantics, we can fix some $n$ and interpret the places in the $n$-tuples representing truth values as the different respects available for a given noun, adjective or verb. The values $1$ and $0$ indicate whether a given object qualifies as having the corresponding noun, adjective or verb in the corresponding respect. Let us fix on the discussion of the example ``John is a man and John is not a man''. For the sake of simplicity, let us suppose we have two respects, the first one is related to being a man in respect to John's gender, the second one is related to being a man as concerned with a given stereotype of masculinity. Then, we have the four options:
\begin{itemize}
    \item $\langle 1,1\rangle$: John is a man according to gender, and according to the stereotype;
    \item $\langle 1,0\rangle$: John is a man according to gender, but not according to the stereotype;
    \item $\langle 0,1\rangle$: John is not a man according to gender, but satisfies the man stereotype;
    \item $\langle 0,0\rangle$: John is not a man according to gender, and also not according to the stereotype.
\end{itemize}
This nicely illustrates the idea that we can have different respects that can be quantified over; basically, for any $n$ we can have a semantics with $n$ respects. It also captures the claim, by \'Egr\'e, that a contradiction is informative when some predicate is not the case for all respects, so that it can be applied in relation to some respects, but not to others. That matches well the idea that if a proposition is the case for all respects (it receives a block of $1$s), then its negation will not be the case (it will receive a block of $0$s). The conjunction the will be just completely false for each respect. There is a sense in which such contradictions say nothing, they exclude the applicability of the predicate according to any respect. In this specific case, there is disagreement as related to every respect.

The distinct consequence relations that can be defined on the top of the Clemens semantics also acquire an interesting reading with that kind of approach. Let us briefly check:
\begin{itemize}
    \item {\bf K3} is the logic obtained when consequence must preserve satisfaction of all the respects; not contradictions allowed, even if informative;
    \item {\bf LP} is the logic obtained when informative contradictions are allowed; uninformative contradictions should be ruled out;
    \item {\bf CL} is the logic where the first respect has a priority over others, so that it is this one that must be preserved.
\end{itemize}

One final point before we leave this particular application. It is interesting to remark once again that a fixed order for the different respects is required, if Clemens semantics is to be used in this case. That means that some regards are considered to be more important than others, at least in each context. This is not completely unrealistic, given that depending on a context, one may privilege some respects as more important than others. In a certain sense, given the classical reading of the consequence relation, privileging the first regard could be read in a kind of epistemic approach to vagueness, where vagueness is only in language, and the first regard is a kind of universal standard (God's knowledge of borders). So, the other regards would play a role similar to the one different agents played in the agent reading.

\subsection{Topic III - Mixed consequence}

Now, let us consider the effect that Clemens semantics may have on topics related to mixed consequence relations. From a technical point of view, such consequence relations are simple to obtain on the top of a three-valued semantics for {\bf K3} and {\bf LP}. However, given that for the cases of these logics such a semantics is typically interpreted in different terms, with the third truth value as meaning either gaps or gluts, respectively, the understanding of the mixed consequence relations seems to face some difficulties: the meaning of the third truth value seems to fluctuate between gap and glut, depending on whether it is taken strictly or tolerantly. Let us be a bit more explicit about it: when the set of truth values is considered from a strict point of view, the third truth value is not designated, and so, is not to be counted as truth, acquiring the features of a gap as per {\bf K3}; when considered from the tolerant point of view, the third truth values is read as per {\bf LP}, looking like a glut. In a sense, then, the third truth values has a kind of chameleon nature. 

In order to face some of such difficulties, a distinct reading for the semantics is offered in \cite{Cobreros2012tolerant}:
\begin{quote}
A second feature of our target semantics is that, while it coincides with the predictions of the many-valued logics LP and K3, it answers to a distinct motivation. Rather than seeing truth as a unified notion to which sentences might answer in three (or more) different ways, our approach posits distinct notions of truth, each of which a sentence may have or fail to have, but none of which is many-valued. \cite[p.365]{Cobreros2012tolerant}
\end{quote}
Although the idea is to provide understanding of the semantic concepts involved, the strategy is requiring that truth be understood as a multiplicity of concepts. That is clearly a very non-classical reading of the notions of truth and falsity, illustrating what one may take as the addition of some new \textit{sui generis} truth values. The result is that those like Haack, who are not sympathetic to such additions of truth values would be intrigued by what `distinct notions of truth' could mean. What Clemens semantics does, in this case, is to provide for chances of uniform readings of the semantic values for mixed consequence cases, just as for standard ones. The intuitive reading advanced by Clemens, or else the agent reading, present before, are nice illustrations. The readings are there before the consequence relation is defined, so they can be used to illuminate the system independently of what kind of consequence relation one plugs in the semantics. 

But besides using the already offered readings for the Clemens semantics for the understanding of the three-valued presentations of mixed consequence relation cases, we can also benefit from those readings to make a sharper sense of what `strict' and `tolerant' mean as per \cite{Cobreros2013reaching}. It is suggested that there are two ways of understanding sentences when it comes to classify them as strict or tolerant: there is a pragmatic way, according to which an \emph{assertion} is qualified in terms of strict and tolerant, and there is an approach through \emph{meaning}, where a sentence may have strict or tolerant meaning. Concerning the pragmatic approach:
\begin{quote}
 [\dots] we can see a direct connection between model-theoretic
value and assertibility. A sentence is either both strictly and tolerantly assertible (value $1$), tolerantly but not strictly assertible (value $\frac{1}{2}$), or not assertible at all (value $0$). We do not allow for sentences that are strictly but not tolerantly assertible; strict assertion, on this picture, is a (strictly) stronger speech act than tolerant assertion. \cite[pp.857-858]{Cobreros2013reaching}
\end{quote}

Notice that as an explanation of what `strict' and `tolerant' mean, those are a bit circular: if we wanted to know what `strict' and `tolerant' mean, the explanation comes in terms already using `strict' and `tolerant'. Those terms gain interesting meanings when one uses the Clemens semantics, and also, one obtains a more fine-grained distinction, allowing a distinction of two kinds of tolerant assertions.\footnote{More kinds are allowed, of course, when ordered $n$-tuples are used, for $2 < n$.} A sentence is strictly assertible, according to the agent reading, if it is asserted by both agents (considering the case of two agents). It is tolerantly assertible, but not strictly assertible, in two distinct scenarios: when only the first agent asserts it, or else when only the second agent asserts it. It is neither strictly nor tolerantly assertible when no agent asserts it. Here, `strict' and `tolerant' qualify the assertions made by each of the agents, which are previously understood in classical terms (thus satisfying Haack's demands). 

The `meaning approach' to strict and tolerant offered by \cite{Cobreros2013reaching} is equally dependent on our having grasped the meaning of `strict' and `tolerant' beforehand: 
\begin{quote}
The other approach works at the level of \textit{meaning}. Rather than supposing that there are two distinct speech acts of assertion, this approach supposes that each sentence has two distinct meanings (or two distinct aspects of its meaning, if you like) that can be asserted: its strict meaning and its tolerant meaning. Understanding meanings as dividing the space of models in two, we can understand a sentence’s strict meaning as one drawing a division between those models on which the sentence takes value $1$ and those on which it takes some value less than $1$, and we can understand a sentence’s tolerant meaning as one drawing a division between those models on which the sentence takes some value greater than 0 and those on which it takes value $0$. \cite[p.858]{Cobreros2013reaching}
\end{quote}
Again, according to Clemens semantics, a sentence may have meaning, only tolerant but not strict, or neither. However, tolerant meaning may be qualified in different guises, just as in the case of tolerant truth. These may be cashed in terms of the agent reading, or of the original reading by Clemens, among others. They do confer a nice illustration of how those notions may be understood in terms of the classical concepts, even though this understanding deviates from the original one proposed by \cite{Cobreros2013reaching}. 

\section{Concluding remarks}\label{sec:conclusion}

In this paper, we have expanded on a semantic framework advanced originally by Matthew Clemens. In particular, we have presented a Clemens semantics for first-order logic, and we also considered the use of such a framework to deal with mixed consequence relations. The benefits of such an investigation were explored through the lenses of a demand formerly expressed by Susan Haack, according to which many-valued logics become more intelligible when additional truth values are analysed in terms of bivalent truth and falsity. We have provided for some readings of Clemens semantics that satisfy such a requirement, and indicated how such readings impact on current attempts to use many-valued logics to deal with some interesting philosophical problems. 

\bibliographystyle{eptcs}
\bibliography{Clemens}
\end{document}